\documentclass[a4paper,USenglish]{lipics-v2018}
\usepackage{microtype}%if unwanted, comment out or use option "draft"

\title{Fast DecreaseKey Heaps with worst-case variants}

\author{Vladan Majerech}{Department of Theoretical Computer Science and Mathematical Logic (KTIML), Charles University, Malostransk\'e n\'am\v est\'\i\ 25, Prague 118 00, Czech Republic}{maj@ktiml.mff.cuni.cz}{ https://orcid.org/0000-0003-3006-2002}{}%mandatory, please use full name; only 1 author per \author macro; first two parameters are mandatory, other parameters can be empty.

\authorrunning{V. Majerech}%mandatory. First: Use abbreviated first/middle names. Second (only in severe cases): Use first author plus 'et al.'

\Copyright{Vladan Majerech}%mandatory, please use full first names. LIPIcs license is "CC-BY";  http://creativecommons.org/licenses/by/3.0/

\subjclass{Information systems$\rightarrow$Information storage systems$\rightarrow$Record storage systems$\rightarrow$Record storage alternatives$\rightarrow$Heap (data structure)}% mandatory: Please choose ACM 2012 classifications from https://www.acm.org/publications/class-2012 or https://dl.acm.org/ccs/ccs_flat.cfm . E.g., cite as "General and reference $\rightarrow$ General literature" or \ccsdesc[100]{General and reference~General literature}. 

%mandatory

\category{}%optional, e.g. invited paper

\relatedversion{}%optional, e.g. full version hosted on arXiv, HAL, or another repository/website

\supplement{}%optional, e.g. related research data, source code, ... hosted on a repository like zenodo, figshare, GitHub, ...

\funding{}%optional, to capture a funding statement, which applies to all authors. Please enter author-specific funding statements as the fifth argument of the \author macro.

\acknowledgements{}%optional

\EventEditors{}
\EventNoEds{1}
\EventLongTitle{}
\EventShortTitle{}
\EventAcronym{}
\EventYear{2020}
\EventDate{}
\EventLocation{}
\EventLogo{}
\SeriesVolume{}
\ArticleNo{1}
\nolinenumbers %uncomment to disable line numbering
\hideLIPIcs  %uncomment to remove references to LIPIcs series (logo, DOI, ...), e.g. when preparing a pre-final version to be uploaded to arXiv or another public repository
\begin{document}
\newcommand{\prevfootnotemark}{\csname @footnotemark\endcsname}
\newcommand{\MakeHeap}{{\bf MakeHeap\/}}
\newcommand{\FindMin}{{\bf FindMin\/}}
\newcommand{\DeleteMin}{{\bf DeleteMin\/}}
\newcommand{\ExtractMin}{{\bf ExtractMin\/}}
\newcommand{\Delete}{{\bf Delete\/}}
\newcommand{\Meld}{{\bf Meld\/}}
\newcommand{\Cut}{{\bf Cut\/}}
\newcommand{\Insert}{{\bf Insert\/}}
\newcommand{\Decrement}{{\bf DecreaseKey\/}}
\newcommand{\Right}{\hbox{\sl right\/}}
\newcommand{\Left}{\hbox{\sl left\/}}
\newcommand{\nill}{\hbox{\sl null\/}}
\newcommand{\Parent}{{\sl parent\/}}
\newcommand{\N}{{\tt N}}
\newcommand{\A}{{\tt A}}
\newcommand{\LL}{{\tt L}}
\newcommand{\Li}{{\tt L}_1}
\newcommand{\Lii}{{\tt L}_2}
\newcommand{\numfootnote}{\footnote}
\global\advance\dimen253 by 5cm
\maketitle

\begin{abstract}
In the paper \cite{ffhwce}, we have described heaps with both \Meld-\Decrement\ and \Decrement\ interfaces, allowing operations with guaranteed worst-case asymptotically optimal times. The paper was intended to concentrate on the \Decrement\ interface, but it could be hard to separate the two described data structures without careful reading. The current paper's goal is not to invent a novel data structure, but to describe a rather easy \Decrement\ version in a hopefully readable form. The paper is intended not to require reference to other papers.
\end{abstract}

\section{Introduction}
Heap is a data structure maintaining a set of elements with keys where keys belong to a linearly ordered universe.
Minimal heap interface supports \Insert\ and \ExtractMin\ methods. \Insert\ adds an element to the structure, \ExtractMin\ returns, and removes the element with the smallest key from the structure. 
In case the structure is empty, \ExtractMin\ returns \nill. 
In the case, there are more elements with the minimal key in the heap, one of them is chosen.

As we cannot distinguish $n!$ possible permutations using just comparisons faster than by $\Omega(\log n!)$, at least one of the operations requires $\Omega(\log n)$ time (we could sort using $n$ {\Insert}s and {\ExtractMin}s).
As for nonempty heap, \ExtractMin\ cannot be used often than \Insert, the interface with $O(1)$ running time for \Insert, and $O(\log n)$ for \ExtractMin, where $n$ denotes the current heap size, is asymptotically optimal.

It could be advantageous to introduce the \FindMin\ method, which returns the minimal heap element and it remains in the heap. 
\ExtractMin\ used to be renamed to \DeleteMin, and it calls \FindMin\ internally.
Optimal interface declares $O(1)$ for \Insert\ and \FindMin, and $O(\log n)$ for \DeleteMin.

\Decrement\ heap interface introduces \Decrement\ method, which allows a decrease of the key of a pointed heap element.
To allow pointing, the \Insert\ method should return a pointer to the element to be used for future references. 
Optimal \Decrement\ heap interface declares $O(1)$ for \Insert, \FindMin\ and \Decrement, and $O(\log n)$ for \DeleteMin.

A lot of implementations of the asymptotically optimal amortized \Decrement\ heap interface are known for a long time.
Recently, several implementations of the asymptotically optimal worst-case \Decrement\ heap interface were published.
In this paper, we present a relatively simple heap with the asymptotically optimal worst-case \Decrement\ heap interface.

\section{Overview}
We have an additional requirement that all keys should be different.
This could be achieved by lexicographical comparison with an added second coordinate. Usually using the binary representation of the pointer to the node as the second coordinate would work.

The base of our heap will be a tree of logarithmic arity which is maintained heap ordered, which means the parent key is always smaller than the child key.
Temporarily during the execution of a heap method, the heap could consist of a list of such trees. 
\FindMin\ method's goal is to convert the forest to one tree and return the root.

The fact that the arity is at most logarithmic is important for \DeleteMin\ method as all root children become a list of heap trees and internal \FindMin\ require time proportional to the number of trees.
To maintain the arity in bounds, we should carefully pick elements to compare. 
This is why the ranks of the nodes are introduced.
Nodes are created with rank 0.
Whenever two nodes of the same rank are compared, an edge is added such that the node with a smaller key becomes the parent of the other and the rank of the parent is incremented.
Most of the comparisons we would make are between nodes of the same rank.
Edges corresponding to such comparisons would be called rank edges.
Unmaintanable ideal heap would have only rank edges. 
In reality, it must be violated.
There will be nonrank edges in the tree created by a comparison of nodes of different ranks.
We call children connected by nonrank edges nonrank roots, the root is considered a nonrank root as well.

Another violation source is the \Decrement. To support it, some nodes could lose some of their children. 
To keep the ranks in logarithmic bounds, we should be careful with such losses. 
We define a loss for a node. Rank roots have loss 0 by definition. 
When a rank edge has been created, the loss of the child is 0. 
Whenever a node loses a rank child, its rank is decremented.
Whenever a node other than a nonrank root has rank decremented, its loss is incremented. 

We introduce a strategy that keeps both the loss violation size (sum of all losses) and the rank-root violation size (number of nonrank roots) logarithmically bounded.
To maintain the violation within limits, we maintain pointers to affected nodes in arrays addressed by their ranks.
Whenever we are ready to insert a second pointer to the same place, we have two nodes of the same rank ready to make a violation reduction.

It could happen, the array becomes almost full, and an invocation of a method could allow a big reduction of the violations even when the worst-case bound does not allow us to perform all allowed reductions.
This is why we maintain stacks of pending pointer insertions and we have to define a strategy of how to stop processing the stacks.
The amortized strategy would just empty the stacks.
We would propose two worst-case strategies, both would maintain both types of violations of at least maximal rank plus one (corresponding to empty stacks and maximally filled arrays).

We will show that the ranks are at most logarithmic in chapter 3. 
Description of violation reduction would follow in chapter 4, the stack reduction strategy analysis will start there and end in chapter 8. Chapters 5, 6, and 7 would concentrate on the implementation details.

\section{Maximal rank}
Let us study the relationship between the number of nodes $n$, the maximal rank $R$, and the maximal total loss $L$.
Let us try to find the minimal possible number of nodes $n(R, L)$, for given $R$ and $L$. 
As we would maintain $L\le R+1$, we are especially interested on $n(R,R+1)$.
The inverse of this function will give us the maximal possible rank $R(n)$ for a heap of $n$ nodes.
A binomial tree with a root of rank $r$ is a tree obtained by adding an edge between the roots of two binomial trees of rank $r-1$.
Such a tree has $2^r$ nodes and could be depicted as a root having roots of binomial trees of ranks $0$ up to $r-1$ as children.
All edges of a binomial tree are rank edges.
Let us call a binomial tree of rank $r$ with losses any tree which could be obtained from the binomial tree of rank $r$ by cutting edges where none cuts a root child and cutting a rank child of a successor is accompanied with the rank decrement of the successor and the loss increase there (there need not be cut at all).
%\looseness=-1

\begin{lemma}
During a history of the heap for any node $x$ of rank $r$ holds, that there is a subtree $B_x$ rooted at $x$ of the full subtree $T_x$ rooted at $x$ such that $B_x$ is a binomial tree of rank $r$ with losses.
\end{lemma}

\begin{proof}
Let us call the stated subtree $B_x$ a binomial guarantee under $x$.
By induction of the heap history. There are two cases. Either the last operation creates an edge or destroys an edge.
Start with the latter case. 

If the destroyed edge is a nonrank edge, it is not part of a binomial guarantee, no rank has changed, so the lemma remains valid.

If the destroyed edge is a rank edge $(x,y)$ with $x$ closer to the root, just all ancestors of $x$ in the rank path of $x$ are affected. Let $x$ has rank $r_x$.
Before the operation, there was a binomial guarantee $B_x$ under $x$.
Let the subtree under vertex $y$ correspond to a binomial subtree of $B_x$ of rank $i$ with losses.
After the cut, the rank of $x$ changes to $r_x-1$.
A child $c$ of $x$ which corresponds to a binomial subtree of $B_x$ of rank bigger than $i$ (with losses) contains a binomial subtree with losses of one smaller rank (just ignore the youngest child of $c$ in $B_x$ prior cut).
This gives us a new binomial subtree with losses of rank $r_x-1$, so the binomial guarantee under $x$ after the cut.
For an ancestor $a$ in the rank path of $x$, different from $x$, let $r_a$ be its rank.
There existed a binomial guarantee $B_a$ under $a$.
If the $(x,y)$ edge is not part of the subtree, the same subtree $B_a$ remains the binomial guarantee under $a$.
Otherwise, cutting a subtree under $(x,y)$ from $B_a$ creates the binomial guarantee under $a$ after the cut (the rank and loss are updated).

Similarly, adding a nonrank edge does not change ranks and guarantees remain valid.

Adding a rank edge between parent $x$ of rank $r$ and a nonrank root $y$ of a rank $\ge r$ increases the rank of $x$ to $r+1$. Binomial guarantee $B_y$ of $y$ includes binomial subtree $B'_y$ of rank $r$ with losses (include just first $r$ branches).
For an ancestor $a$ of $x$ in the rank path of $x$, including $x$, the binomial guarantee under $a$ is extended by $(x,y)$ and $B'_y$ (the rank and loss of $x$ are updated).
\end{proof}

Thanks to the lemma, $n(R, L)$ must be a binomial tree of rank $R$ with losses.
Our goal is to cut $L$ edges from the binomial tree of rank $R$, not decreasing the rank such that the tree becomes as small as possible.
Cutting children is forbidden as it would decrease the rank.
Cutting a deeper successor than a grandchild is ineffective as cutting the grandchild on the path to the root reduces the tree size more.
The most effective strategy is to cut the grandchildren whose subtrees are the largest.
There are $k$ grandchildren of rank $R-1-k$ so of the size $2^{R-1-k}$.
There are $\sum_{i=1}^{k} i={k+1 \choose 2}$ grandchildren of rank at most $R-1-k$.
If $L<{k+1 \choose 2}$, we know no grandchild of rank $R-2-k$ was cut in the smallest tree of rank $R$ and the total loss $L$ and therefore $n(R,L)>(k+1)\cdot 2^{R-2-k}$.

We are especially interested in the case $n(R,R+1)$, so $R+1<(k^2+k)/2$.
This gives us $k^2+k-2R-2>0$. Higher root of this polynomial is $\sqrt{2R+9/4}-1/2$ so for $k\ge \sqrt{2R+9/4}+1$ we have $n(R,R+1)>(k+1)\cdot 2^{R-2-k}\ge(2+\sqrt{2R+9/4})\cdot 2^{R-3-\sqrt{2R+9/4}}\in \Omega(2^{R(1-o(1))})$.
Our goal is to find a nice function that will still be a lower bound for $n(R, R+1)$.
The numerical evaluation clearly shows $2^{(R-4)/1.2}$ is for all integers smaller than the mentioned bound and therefore $R(n)<4+1.2\log_2 n$.
%The bounds could be probably slightly improved considering the loss violation rank distributions which could appear according to the chosen violation reduction strategy.

\section{Violation reductions}
\vbox to 60mm{
\hbox{\kern17mm\pdfximage width 14cm {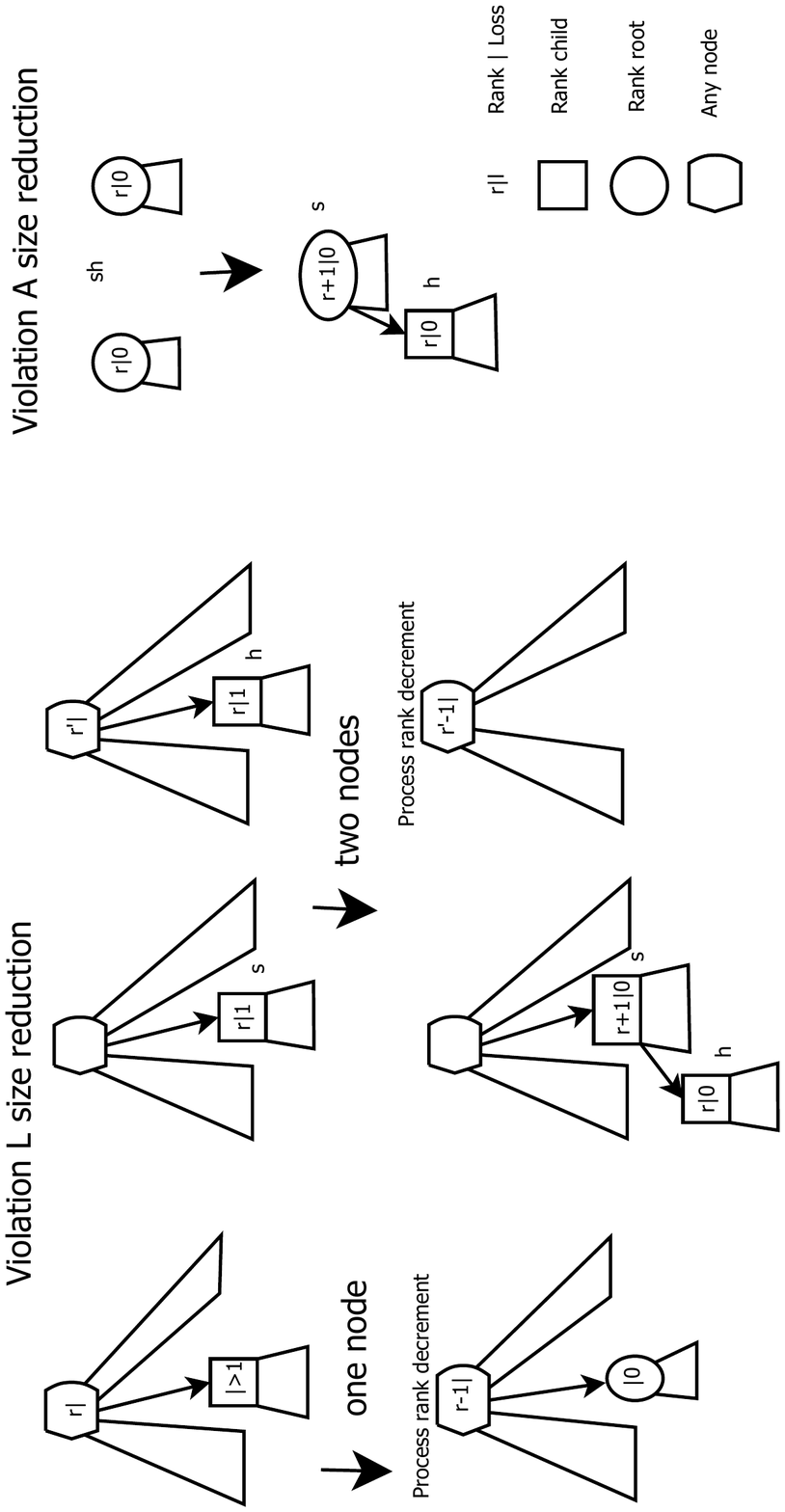}
\rlap{\smash{\pdfsave\pdfsetmatrix{0 -0.56 0.56 0}
\pdfrefximage\pdflastximage}}\pdfrestore
\hss}
\vss
\hbox{Figure 1: Nontrivial cases of reductions to maintain the heap shape}
\kern4mm
}

\begin{table}
 \begin{center}
  \caption{Effect of stacks reductions $\Phi_A=|R_A|+2|C_A|$, $\Phi_L=3|R_L|+4|C_L|$, $\Phi=\Phi_A+\Phi_L$}
	\label{tab:effNoDeffered}
	\begin{tabular}{lrrrrrrrr}
	 Reduction & $|R_A|$ & $|C_A|$ & $\Phi_A$ & $|R_L|$ & $|C_L|$ & $\Phi_L$ & $\Phi$ & $p$\\
	 \hline
	 \hline
	 $|C_A|$ type $\not=\A$& $0$ & $-1$ & $-2$ & $0$ & $0$ & $0$ & $-2$ & $0$ \\
	 \hline
	 $|C_A|$ type $\A$ no match& $+1$ & $-1$ & $-1$ & $0$ & $0$ & $0$& $-1$ & $1$ \\
	 \hline
	 $|C_A|$ type $\A$ matched& $-1$ & $0$ & $-1$ & $0$ & $0$ & $0$ & $-1$ & $2$\\
	 \hline
	 $|C_L|$ type $\not=\LL$ & $0$ & $0$ & $0$ & $0$ & $-1$& $-4$ & $-4$ & $0$ \\
	 \hline
	 $|C_L|$ subtype $\Lii$ & $\le 0$ & $\le +2$ & $\le +3$& $\le 0$ & $\le 0$ & $\le -3$ & $\le -1$ & $\le 3$\\
	 \ - parent $\A$ in $R_A$& $-1$ & $+2$ & $+3$ & $0$ & $\le -2$ & $\le -8$ & $\le -5$ & $3$ \\
	 \ - parent $\A$ in $C_A$& $0$ & $+1$ & $+2$  & $0$ & $\le -2$ & $\le -8$ & $\le -6$ & $1$ \\
	 \ - parent $\Li$ in $R_L$& $0$ & $+1$ & $+2$ & $-1$ & $\le 0$ & $\le -3$ & $\le -1$ & $3$ \\
	 \ - parent $\LL$ in $C_L$& $0$ & $+1$ & $+2$ & $0$ & $\le -1$ & $\le -4$& $\le -2$ & $1$ \\
	 \ - parent $\N$& $0$ & $+1$ & $+2$ & $0$ & $\le -1$ & $\le -4$ & $\le -2$ & $2$ \\
	 \hline
	 $|C_L|$ subtype $\Li$ no match & $0$ & $0$ & $0$ & $+1$ & $-1$& $-1$ & $-1$ & $1$ \\
	 \hline
	 $|C_L|$ subtype $\Li$ matched& $\le 0$ & $\le +1$ & $\le +1$ & $\le -1$ & $\le +1$ & $\le -2$ & $\le -2$ & $\le 3$ \\
	 \ - parent of $h$ $\A$ in $R_A$& $-1$ & $+1$ & $+1$& $-1$ & $-1$ & $-7$ & $-6$ & $3$ \\
	 \ - parent of $h$ $\A$ in $C_A$& $0$ & $0$ & $0$ & $-1$ & $-1$ & $-7$ & $-7$ & $1$ \\
	 \ - parent of $h$ $\Li$ in $R_L$& $0$ & $0$ & $0$ &$-2$ & $+1$ & $-2$ & $-2$ & $3$ \\
	 \ - parent of $h$ $\LL$ in $C_L$& $0$ & $0$ & $0$ & $-1$ & $0$ & $-3$ & $-3$ & $1$ \\
	 \ - parent of $h$ $\N$& $0$ & $0$ & $0$ & $-1$ & $0$ & $-3$ & $-3$ & $2$ \\
	 \hline
	\end{tabular}
 \end{center}
 \vskip 4pt
 Here, $p$ denotes the number of pointer changes not reflected in heap trees during reduction.
 We can see that each stack reduction decrements $\Phi$ by at least 1.
\end{table}

Each node remembers its violation type, which is either $\A$, $\LL$, or $\N$ otherwise.

Rank roots would be maintained as violations of type $\A$. 
Violations of type $\A$ would be pushed to the stack $C_A$ and from the stack popped to the array $R_A$ addressed by rank.

During a $|C_A|$ reduction, if the popped stack node $x$ is already not of type $\A$, the stack item is just discarded.
The stack item is just discarded as well if $R_A$ for $x$'s rank points to $x$.
Another case is $R_A$ for $x$'s rank contains \nill, then the pointer to $x$ is stored there, and the $|C_A|$ reduction step ends.
Last, and the most important case is when $R_A$ for $x$'s rank points to another violation $y$ of the same rank and actual violation type, the $\A$ reduction step could be applied after putting \nill\ to $R_A$ of $x$'s rank. 
Violation reduction step of type $\A$ links nodes $x$, $y$ of the same rank. 
(Their keys are compared, let node $s$ be the one with a smaller key while $h$ the other. We cut $h$ from its parent (if there exists a nonrank edge) and put it as a rank child of $s$. This increases the rank of $s$.
Node $h$ violation type is changed to $\N$, node $s$ is added to $C_A$.)

All nodes with a nonzero loss would be maintained as violations of type $\LL$. 
This type has a subtype $\Li$ for nodes with loss exactly 1 and a subtype $\Lii$ for nodes with loss at least 2.
Violations of type $\LL$ would be pushed to the stack $C_L$, from which the popped violations of subtype $\Li$ will be inserted to the array $R_L$ addressed by rank.
The symbol $|C_L|$ has a weighted meaning. Weight of nodes of subtype $\Lii$ corresponds to the loss of the node, while the weights of other nodes are 1 (including nodes of violation type other than $\LL$).

Similarly, as for $|C_A|$ reduction, when during $|C_L|$ reduction the popped stack node $x$ is already not of type $\LL$ or $R_L$ of $x$'s rank points to $x$, the stack item is just discarded.
Different is the second case when $x$'s subtype is $\Lii$. 
It invokes one node loss reduction, which takes node $x$ with loss at least 2, the reduction makes $x$ a nonrank child of its parent $p$.
This creates a new rank root $x$, so $x$ is put to $C_A$ and the violation type of $x$ is changed to $\A$.
The rank of $p$ is decremented.
Unless the violation type of $p$ is $\N$, $p$ should be removed from the array identified by its type (If there is the \nill\ there, we know $p$ is on the stack).
If $p$ is a rank child, it should be inserted to $C_L$ (if it is not there), and its type changed to $\LL$, its subtype should be set to reflect the loss increase. Total loss was reduced by at least 1.
If $p$ is a rank root, it should be pushed to the stack $C_A$ unless it already resists there.
The third case of $|C_L|$ reduction is for $\Li$ subtype when the array $R_L$ of $x$'s rank contains \nill.
As for $|C_A|$ reduction, the pointer to $x$ is stored in $R_L$ and the $|C_L|$ reduction step ends.
The last case is when $R_L$ for $x$'s rank (for a node of subtype $\Li$) points to another violation $y$ of the same rank and the violation type $\LL$, the reduction step could be applied after putting \nill\ to $R_L$ of $x$'s rank. 
The violation reduction step of type $\LL$, for nodes $x$, $y$ of equal rank, and loss 1, links the two nodes. 
(Their keys are compared, let $h$ and $s$ be the nodes with higher and smaller keys, respectively. 
 Remove $h$ from its parent and link it under $s$ by a rank edge. 
 This reduces the loss of $s$ to 0 and sets the loss of $h$ to 0, so both $s$ and $h$ violation types are changed to $\N$.
 Original parent $p$ of $h$ decrements rank by 1. 
 Unless the violation type of $p$ is $\N$, 
 $p$ should be removed from the array identified by its type 
 (if there is the \nill\ there, we know $p$ is on the stack).
  If $p$ is a rank child, its type should be changed to $\LL$ and $p$ inserted to $C_L$ (if it is not there),
  the subtype of $p$ should be set to reflect the loss increase. Total loss was reduced by at least 1.
  If $p$ is a rank root, it should be inserted to $C_A$ (if it is not there).)
\looseness=1

In Figure 1 you can see the reductions and in Table \ref{tab:effNoDeffered} you can see the effect of different reductions and the definition of $\Phi_A$, $\Phi_L$, and $\Phi$.

Both worst-case stack reduction strategies would empty the stacks after each \DeleteMin\ operation (the amortized version of \DeleteMin\ is always used).
The first worst-case strategy calculates the differences $\Delta \Phi_A$ and $\Delta \Phi_L$ happening from the start of the invoked method (for a method other than \DeleteMin\ they will be bounded by a constant). While any of them is positive and the corresponding stack is nonempty, the corresponding stack size reduction is invoked. As each reduction decreases $\Phi$ by at least 1 and each coordinate change is bounded by a constant, the number of reduction steps is therefore bounded by a constant.
The second worst-case strategy calculates the maximal number of reduction steps required by the first strategy for each method, plans this number of reduction steps, and does the planned reduction steps unless each stack of an unfinished plan is empty.
The situation is simplified as $|C_A|$ reductions do not change $\Phi_L$, so we can make $|C_L|$ reductions first and finish with $|C_A|$ reductions.

Both strategies have in common that either the corresponding stack $C_X$ is empty or the corresponding $\Phi_X$ is at most as big as the last time the stack was empty and the current $n$ did not decrease from $n'$ at that time.
In the former case, all violations of a given type $X$ are addressed in the array by at most one pointer per rank of rank 0 up to $R(n)$, so there are at most $R(n)+1$ violations.
From the form of $\Phi_X$ ($x>0$), we get in the latter case $x(R(n')+1)\ge x|R_X|+(x+1)|C_X|\ge x|R_X|+x|C_X|$ bounding $|R_X|+|C_X|$ by $R(n')+1\le R(n)+1$ as well. Actually, the violation size could be smaller than $|R_X|+|C_X|\le R(n)+1$ in case a violation is pointed more times in $R_X\cup C_X$ or if there is a pointer in $C_X$ to a node which is not a violation of the type $X$.

We will return to the second strategy after presenting the implementation details.

\section{Heap structure}

Heap information contains a pointer to a list of heap tree roots, an array of pointers to four arrays $R[A]$, $R[L]$, $C[A]$, and $C[L]$ and stack pointer indices $P[A]$, $P[L]$ initialised to 0.
The pointed arrays are expected to have sufficient size, but standard worst-case array doubling could be used to solve a problem with maximal rank exceeding the planned value.
The pointed arrays are filled with {\nill}s, the pointer to the list of heap tree roots is initialized with \nill.

Heap node contains an element with key, integer rank, violation subtype (either $\A$, $\Li$, $\Lii$, or $\N$), a pointer to the parent (which is \nill\ for a heap tree root), and pointers \Left\ and \Right\ to siblings. 

The list of heap tree roots uses sibling pointers maintained in the heap nodes.
The parent edge of node $x$ is a rank edge whenever the violation type of $x$ is not $\A$. 
The violation type is maintained implicitly using the violation subtype\footnote{We need loss of a node only in the analysis, the subtype is all we need in the implementation.}.

Left pointers in sibling lists are maintained cyclic (\Left\ of the leftmost node points to the rightmost node), while right pointers are maintained acyclic (\Right\ of the rightmost contains \nill). Except for the first node of a sibling list $x\to\Left\to\Right=x$.
This allows access of both ends in constant time as well as adding or removing of a given node. 

\section{Private methods}
We will describe the public methods using private methods. Their use could be slightly optimized\footnote{In few places $\{\}$ mark duplicated work.}. Decomposition into private methods makes the description easier.

Type($S$): Tabulated conversion of subtype to type by prescription $(\A\to\A),(\N\to\N),(\Li\to\LL),(\Lii\to\LL)$.

SetViolationSubtype($S$) for a node $x$: Unless the original violation subtype $O$ of $x$ is $\N$, the pointer to $x$ is removed from the corresponding array $R[{\rm Type}(O)]$ (addressed by $x$'s rank) if there is anything else than the pointer to $x$, we know $x$ remains in $C[{\rm Type}(O)]$.
Then we change the violation subtype of $x$ to $S$ and insert $x$ to the stack $C[{\rm Type}(S)]$ (unless $S=\N$ or we know the node is already there). Not all combinations of $O$ with $S$ are used.

DecrenmentRank() for a node $x$: 
Unless violation subtype $O$ of $x$ is $\Lii$, let $S$ be obtained from $O$ by tabulated conversion by prescription
 $(\A\to\A),(\N\to\Li),(\Li\to\Lii)$, and SetViolationSubtype($S$) will be called for $x$. 
Then in all cases, the rank of the node $x$ is decremented.

CutFromParent() for a node $c$:
Unless the original violation subtype $O$ of $c$ is $\A$, DecrementRank() is called for the parent of $c$.
In all cases, $\{$parent pointer of $c$ would be set to \nill\ and$\}$ $c$ would be removed from its sibling list.
The calling method should afterward add $c$ to another sibling list. 
The method could be called even when $c$ has no parent.

\begin{table}
 \begin{center}
  \caption{Effect of private methods on violations}
	\label{tab:privblockNoMeld}
	\begin{tabular}{lrrrr}
	 Method & $\Phi_A$ & $\Phi_L$ & $\Phi$ & $p$\\
	 \hline
	 \hline
	 SetViolationSubtype $\A$& $\le +2$ & $\le 0$ & $\le +2$ & $\le 2$ \\
	 \ | from $\A$ & $\le +1$ & $0$ & $\le +1$ & $\le 2$ \\
	 \ | from $\LL$ & $+2$ & $\le 0$ & $\le +2$ & $\le 2$ \\
	 \ | from $\N$ &  $+2$ & $0$ & $+2$ & $1$ \\
	 SetViolationSubtype $\Li$ (only from $\N$) &  $0$ & $+4$ & $+4$ & $1$ \\
	 SetViolationSubtype $\Lii$ (only from $\Li$) & $0$ & $\le +5$ & $\le +5$ & $\le 2$ \\
	 SetViolationSubtype $\N$ & $\le 0$ & $\le 0$ & $\le 0$ & $\le 1$ \\
	 \ | from $\A$ &  $\le 0$ & $0$ & $\le 0$ & $\le 1$ \\
	 \ | from $\LL$ & $0$ & $\le 0$ & $\le 0$ & $\le 1$ \\
	 DecrementRank& $\le +1$ & $\le +5$ & $\le +5$ & $\le 2$ \\
	 \ | $\A$ & $\le +1$ & $0$ & $\le +1$ & $\le 2$ \\
	 \ | $\N$ & $0$ & $+4$ & $+4$ & $1$ \\
	 \ | $\Li$ & $0$ & $\le +5$ & $\le +5$ & $\le 2$ \\
	 \ | $\Lii$ & $0$ & $+4$ & $+4$ & $0$ \\
	 CutFromParent& $\le +1$ & $\le +5$ & $\le +5$ & $\le 2$ \\
	 Link & $\le +1$ & $\le +5$ & $\le +5$ & $\le 4$ \\
	 \ + CutFromParent for $h$& $\le +1$ & $\le +5$ & $\le +5$ & $\le 2$ \\
	 \ + SetViolationSubtype($H$) of $h$& $\le 0$ & $\le 0$ & $\le 0$ & $\le 1$ \\
	 \ + SetViolationSubtype($S$) of $s$& $\le +1$ & $\le 0$ & $\le +1$ & $\le 2$ \\
	 Link of rank roots of the same rank & $\le +1$ & $0$ & $\le +1$ & $\le 3$ \\
	 Link of rank roots of different rank & $\le +1$ & $0$ & $\le +1$ & $\le 3$ \\
	 Link of $\Li$ nodes of the same rank & $\le +1$ & $\le +5$ & $\le +5$ & $\le 4$ \\
	 \hline
	\end{tabular}
 \end{center}
 \vskip 4pt
 Here, $p$ again denotes the number of pointer changes not reflected in heap trees.
\end{table}

Link($x$,$y$):
There could be asserted that the violation subtype of $x$ and $y$ should be the same and from $\{\A,\Li\}$, the ranks of $x$ and $y$ should be equal in the case $\Li$.
The keys of $x$ and $y$ are compared, let node $s$ be the one with a smaller key while $h$ the other. 
CutFromParent() for $h$ is called and $h$ is added as the leftmost child of $s$, the parent pointer of $h$ is set to $s$.
If $rank(s)\le rank(h)$, let $H=\N$, otherwise let $H=\A$. 
SetViolationSubtype($H$) for $h$ is called unless the violation subtype of $h$ is $H$.
If $H=\N$, let $S$ be obtained from the subtype of $s$ by tabulated conversion by prescription
 $(\A\to\A),(\Li\to\N)$, SetViolationSubtype($S$) is called for $s$ and the rank of $s$ is incremented.

Pushing and Poping on the stack $C[X]$ using the stack pointer index $P[X]$ is standard (left as an exercise).

Let us repeat the nontrivial cases of different stack size reductions using private methods.
$|LC|$ reduction for $x$ of subtype $\Lii$ just calls SetViolationSubtype($\A$) for $x$ and Dec\-re\-ment\-Rank() for its parent.
$|LC|$ reduction for $x$, $y$ of subtype $\Li$ and the same rank just calls Link($x$,$y$).
$|AC|$ reduction for $x$, $y$ of subtype $\A$ and the same rank just calls Link($x$,$y$).

\section{Public methods}
\begin{table}
 \begin{center}
  \caption{Effect of public methods on violations}
	\label{tab:pubmethNoMeld}
	\begin{tabular}{lrrrr}
	 Method & $\Phi_A$ & $\Phi_L$ & $\Phi$ & $p$\\
	 \hline
	 \hline
	 \Insert & $\le +3$ & $0$ & $\le +3$ & $\le 3$\\
	 \ + \FindMin\ phase 0 &  $+2$ & $0$ & $+2$ & $1$\\
	 \ + \FindMin\ phase 2 & $\le +1$ & $0$ & $\le +1$ & $\le 2$\\
	 \hline
	 \FindMin & $0$ & $0$ & $0$ & $0$ \\
	 \hline
	 \DeleteMin \\
	 \ + SetViolationSubtype($\N$) of $\rho$ & $\le 0$ & $\le 0$ & $\le 0$ & $\le 1$ \\
	 \ + \FindMin\ phase 0 & $\le 2R(n)+2$ & $0$ & $\le 2R(n)+2$& $\le R(n)+1$ \\
   \ + \FindMin\ phase 2 & $\le 1$ & $0$ & $0$& $\le \lceil 3(R(n)+1)/2\rceil$ \\
	 \hline
	 \Decrement & $\le +4$ & $\le +5$ & $\le +8$ & $\le 6$ \\
	 \ + CutFromParent() for $x$ & $\le +1$ & $\le +5$ & $\le +5$ & $\le 2$ \\
	 \ + \FindMin\ phase 0 & $\le +2$ & $\le 0$ & $\le +2$ & $\le 2$\\
	 \ + \FindMin\ phase 2 & $\le +1$ & $0$ & $\le +1$ & $\le 2$\\
	 \hline
	\end{tabular}
 \end{center}
 \vskip 4pt
 Here, $p$ again denotes the number of pointer changes not reflected in heap trees. 

 \DeleteMin\ requires at most $2R(n)+3\le 2.4\log_2(n)+11$ stack size reductions in the amortized sense. 
  There is at most 3 pointer overhead per stack size reduction, so altogether it would generate at most 
	$\lceil (17R(n)+25)/2\rceil\le 10.2\log_2 n+47$ pointer overhead in the amortized sense.
  If $\Phi^0$ be potential before and $\Phi^E$ after \DeleteMin, we should include the difference into account as well. 
	But $\Phi^0\le 4R(n)+4$ and $\Phi^E\ge 0$. Therefore, we have the worst case bound $6R(n)+7\le 7.2\log_2(n)+31$ stack reductions. 
	In the worst case, there will be at most $\lceil (41R(n)+49)/2\rceil\le 24.6\log_2 n+107$ pointer overhead.
\end{table}

\Insert($k$)\ method creates a new node $x$ with violation subtype $\N$, key $k$, rank 0, $\{$no parent$\}$ and no child. 
It adds $x$ as a new root to the list of heap tree roots and it invokes \FindMin.
It returns $x$ for future references. 

\FindMin\ method in phase 0 traverses the list of roots and it sets their parent pointers to \nill, it calls SetViolationSubtype($\A$) for each of them which does not have subtype $\A$. In phase 1, the stack size reductions are made.
In the amortized variant, they are performed until the stacks are empty.
In the worst-case variant, the reductions are applied according to the strategy
(the 1st strategy makes $\Delta\Phi_X$ nonpositive or stack $C[X]$ empty, the 2nd strategy makes enough reductions to be sure the same holds)\footnote{The analysis would give the same asymptotic complexity if phase 1 is omitted, but making $|C_A|$ reductions early is an attempt to create rank edges directly rather than converting created nonrank edges by $|C_A|$ reductions later. In case only {\Insert} methods are invoked, it reduces the number of comparisons by a factor of about 3/4.}.
In phase 2, it traverses the heap tree roots left wise, linking two neighboring roots interlaced with steps to the left in the circular list (to link the roots as even as possible).
We finish when only one tree remains. Its root points to the minimum and it will be returned.
In phase 3, which is last, the stack size reductions are made again.
The amortized variant works until the stacks are empty, while worst-case variants according to the strategy to reflect the real $\Delta\Phi_A$ change during the second phase (or the maximal possible $\Delta\Phi_A$ change)\footnote{Skipping phase 3 would complicate the arguments, but doing all reductions in phase 1 is an alternative.}.

\DeleteMin\ method implements only the amortized variant, which has a guaranteed worst-case time $O(\log n)$, so no maintenance of $\Delta\Phi$ coordinates is needed. 
Let $\rho$ be the only tree root. It updates the pointer to the list of roots to point to the leftmost child of $\rho$. 
SetViolationSubtype($\N$) is called for $\rho$. At the end, \FindMin\ is called and $\rho$ is discarded\footnote{After phase 1 of FindMin the stacks are empty and no pointer to $\rho$ in the structure remains.}. 

\Decrement($x$, $k$) method calls CutFromParent() for $x$ and $x$ is added as a new root to the list of heap tree roots.
Then in all cases, it updates the key at node $x$ to $k$.
It invokes \FindMin\ at the end.

\section{Second worst-case strategy}
We can see the effect of public methods on $\Phi$ coordinates and pointer overhead in the table \ref{tab:pubmethNoMeld}. 
The \DeleteMin\ analysis is part of the table.

For the amortized version of \Decrement, we got stack-size reduction requirements by at most $8$ and $6$ in additional pointer overhead, and for the amortized version of \Insert\ we got stack-size reduction requirements by at most $3$ and $3$ in additional pointer overhead. Including stack-size reductions, this makes the amortized pointer overhead at most $30$ per \Decrement\ and $12$ per \Insert.

Only \Decrement\ can increase the $\Phi_L$. It is increased by at most $5$. Each $|C_L|$ reduction reduces $\Phi_L$ by at least 1, so $5$ $|C_L|$ reduction steps when \Decrement\ calls \FindMin\ are sufficient to maintain $\Phi_L$ in bounds.
The largest increase of $\Phi_A$ per $|C_L|$ reduction is by $3$, so $5$ reduction steps by the second strategy increase $\Phi_A$ by at most $15$. Therefore, $\Phi_A$ could increase by $18$ before $|C_A|$ reductions start at \FindMin\ phase 1. Each reduction decreases $\Phi_A$ by at least 1, so $18$ $|C_A|$ reductions in phase 1 are sufficient. At most $1$ additional $|C_A|$ reduction should be performed at phase 3. The stack size reductions generate pointer maintenance overhead at most $53$, so the pointer maintenance overhead for \Decrement\ according to the second worst-case strategy is at most $59$.

A careful look at the first strategy will show the maximal increase of $\Phi_A$ would be $5$ obtained by the first reduction increasing $\Phi_A$ by $2$ and the second by $3$, while $\Phi_L$ reaches a smaller value than before the \Decrement\ started.
\Decrement\ increases $\Phi_A$ by at most $7$ before $|C_A|$ reductions start at \FindMin\ phase 1 (CutFromParent does not increase $\Phi_A$ and $\Phi_L$ simultaneously) so at most $7$ $|C_A|$ reductions are performed in phase 1 of \FindMin\ according to the first strategy, and at most $1$ $|C_A|$ reduction is performed at phase 3. The stack size reductions generate the pointer maintenance overhead at most $22$, so the pointer maintenance overhead for \Decrement\ according to the first worst-case strategy is at most{\penalty1000} $28$.\footnote{Third worst-case strategy would calculate $\Phi_L$ reduction made and stop at $\ge 5$ , then $7+1$ $|C_A|$ reductions suffice.}

The situation with stack size reductions is much easier for \Insert.
There are no $|C_L|$ reductions at all, and at most $2$ $|C_A|$ reductions in phase 1 of \FindMin, and at most $1$ $|C_A|$ reduction in phase 3. This generates (in both worst-case strategies) at most $6$ pointer maintenance overhead, so the worst-case \Insert\ has at most $9$ pointer maintenance overhead in total.

Let us repeat at the end that the second worst case strategy plans for \FindMin\ called from \Decrement\ $5$ $|C_L|$ and $18$ $|C_A|$ reductions to phase 1, and $1$ $|C_A|$ reduction to phase 3.
It plans for \FindMin\ called from \Insert\ $2$ $|C_A|$ reductions to phase 1, and $1$ $|C_A|$ reduction to phase 3.

%%
%% Bibliography//
%%

%% Please use bibtex, 
\bibliography{fdk}
\end{document}